\documentclass[11pt]{article}
\usepackage{times}
\usepackage{amsmath, amssymb, amsthm}
\usepackage{graphicx}
\usepackage{algorithmic}
\usepackage{algorithm}
\usepackage{wrapfig}
\usepackage{color,ifpdf}

\advance \topmargin by -\headheight
\advance \topmargin by -\headsep

\evensidemargin \oddsidemargin
\definecolor{Darkblue}{rgb}{0,0,0.4}
\definecolor{Brown}{cmyk}{0,0.81,1.,0.60}
\definecolor{Purple}{cmyk}{0.45,0.86,0,0}
\newcommand{\mydriver}{hypertex}
\ifpdf
 \renewcommand{\mydriver}{pdftex}
\fi
\usepackage[breaklinks,\mydriver]{hyperref}
\hypersetup{colorlinks=true,
            citebordercolor={.6 .6 .6},linkbordercolor={.6 .6 .6},%
citecolor=Darkblue,urlcolor=black,linkcolor=Darkblue,pagecolor=black}
\newcommand{\lref}[2][]{\hyperref[#2]{#1~\ref*{#2}}}

\newcommand{\OPT}{{\bf OPT}}
\newcommand{\E}{\mathrm{E}}
\newcommand{\I}{\mathcal{I}}
\newcommand{\stream}{\pmb{\sigma}}

\newcommand{\bR}{\mathbb{R}}

\newcommand{\FMV}{{\sf FMV}}
\newcommand{\sse}{\subseteq}

\newtheorem{theorem}{Theorem}[section]

\newtheorem{observation}[theorem]{Observation}
\newtheorem{lemma}[theorem]{Lemma}
\newtheorem{proposition}[theorem]{Proposition}

\newcommand\ignore[1]{}
\newcommand\tsty{\textstyle}

\newcommand{\bZ}{\mathbb{Z}}
\newcommand{\F}{\mathcal{F}}
\newcommand{\rme}{\mathrm{e}}

\theoremstyle{remark}

\newtheorem{claim}[theorem]{Claim}

\newcounter{note}[section]

\newcommand{\initOneLiners}{%
    \setlength{\itemsep}{0pt}
    \setlength{\parsep }{0pt}
    \setlength{\topsep }{0pt}
}
\newenvironment{OneLiners}[1][\ensuremath{\bullet}]
    {\begin{list}
        {#1}
        {\initOneLiners}}
    {\end{list}}

\theoremstyle{definition}
\newtheorem{definition}[theorem]{Definition}
\title{Constrained Non-Monotone Submodular Maximization: \\Offline and
  Secretary Algorithms}

\author{
Anupam Gupta \thanks{Carnegie Mellon University, Pittsburgh, PA. Email: {\tt anupamg@cs.cmu.edu}}
\and
Aaron Roth \thanks{Microsoft Research New England, Cambridge, MA. Email: {\tt alroth@cs.cmu.edu}}
\and
Grant Schoenebeck \thanks{UC Berkeley, Berkeley, CA. Email: {\tt grant@cs.berkeley.edu}}
\and
Kunal Talwar \thanks{Microsoft Research, Mountain View, CA. Email: {\tt kunal@microsoft.com}}
}

\date{}

\begin{document}

\maketitle
\thispagestyle{empty}
\begin{abstract}
  Constrained submodular maximization problems have long been studied, most recently in the context of auctions and computational advertising, 
  with near-optimal results known under a variety of constraints when
  the submodular function is \emph{monotone}. The case of non-monotone
  submodular maximization is less well understood: the first approximation
  algorithms even for the unconstrained setting were given by Feige et
  al.\ \emph{(FOCS '07)}. More recently, Lee et al.\ \emph{(STOC '09,
    APPROX '09)} show how to approximately maximize non-monotone submodular
  functions when the constraints are given by the intersection of $p$
  matroid constraints; their algorithm is based on local-search
  procedures that consider $p$-swaps, and hence the running time may be
  $n^{\Omega(p)}$, implying their algorithm is polynomial-time only for
  constantly many matroids.

  \medskip In this paper, we give algorithms that work for
  \emph{$p$-independence systems} (which generalize constraints given by
  the intersection of $p$ matroids), where the running time is
  $\text{poly}(n,p)$. Both our algorithms and analyses are simple: our
  algorithm essentially reduces the non-monotone maximization problem to
  multiple runs of the greedy algorithm previously used in the monotone
  case. 
  Our idea of using existing algorithms for monotone functions to solve
  the non-monotone case also works for maximizing a submodular function
  with respect to \emph{a knapsack constraint}: we get a simple
  greedy-based constant-factor approximation for this problem.



  \medskip With these simpler algorithms, we are able to adapt our
  approach to constrained non-monotone submodular maximization to the
  \emph{(online) secretary setting}, where elements arrive one at a time
  in random order, and the algorithm must make irrevocable decisions
  about whether or not to select each element as it arrives. We give
  constant approximations in this secretary setting when the algorithm
  is constrained subject to a uniform matroid or a partition matroid,
  and give an $O(\log k)$ approximation when it is constrained by a
  general matroid of rank~$k$.
\end{abstract}
\thispagestyle{empty}
\setcounter{page}{0}
\clearpage

\section{Introduction}

We present algorithms for maximizing (not necessarily monotone)
non-negative submodular functions satisfying $f(\emptyset) = 0$ under a
variety of constraints considered earlier in the literature. Lee et
al.~\cite{LMNS-journal, LeeSV09} gave the first algorithms for these problems
via local-search algorithms: in this paper, we consider greedy
approaches that have been successful for \emph{monotone} submodular
maximization, and show how these algorithms can be adapted very simply
to non-monotone maximization as well. Using this idea, we show the
following results:
\begin{itemize}
\item We give an $O(p)$-approximation for maximizing submodular
  functions subject to a $p$-independence system. This extends the
  result of Lee et al.~\cite{LMNS-journal,LeeSV09} which applied to
  constraints given by the intersection of $p$ matroids, where $p$ was a
  constant. (Intersections of $p$ matroids give $p$-indep.\ systems, but
  the converse is not true.) Our greedy-based algorithm has a run-time
  polynomial in $p$, and hence gives the first polynomial-time
  algorithms for non-constant values of~$p$.

\item We give a constant-factor approximation for maximizing submodular
  functions subject to a knapsack constraint. This greedy-based
  algorithm gives an alternate approach to solve this problem; Lee et
  al.~\cite{LMNS-journal} gave LP-rounding-based algorithms that achieved a
  $(5 + \epsilon)$-approximation algorithm for constraints given by the
  intersection of $p$ knapsack constraints, where $p$ is a constant.
\end{itemize}
Armed with simpler greedy algorithms for nonmonotone submodular
maximization, we are able to perform constrained nonmonotone submodular
maximization in several special cases in the secretary setting as well:
when items arrive online in random order, and the algorithm must make
irrevocable decisions as they arrive.
\begin{itemize}
\item We give an $O(1)$-approximation for maximizing submodular
  functions subject to a cardinality constraint and subject to a
  partition matroid. (Using a reduction of~\cite{BDGIT09}, the latter
  implies $O(1)$-approximations to e.g., graphical matroids.) Our
  secretary algorithms are simple and efficient.

\item We give an $O(\log k)$-approximation for maximizing submodular
  functions subject to an arbitrary rank $k$ matroid constraint. This
  matches the known bound for the \emph{matroid secretary problem}, in
  which the function to be maximized is simply linear.
\end{itemize}
No prior results were known for submodular maximization in the secretary
setting, even for \emph{monotone} submodular maximization; there is some
independent work, see \S\ref{sec:related-indep} for details.

Compared to previous offline results, we trade off small constant factors in our approximation ratios of our algorithms for exponential improvements in run time:
maximizing nonmonotone submodular functions subject to (constant) $p \geq 2$
matroid constraints currently has a $(\frac{p^2}{p-1} + \epsilon)$
approximation due to a paper of Lee, Sviridenko and
Vondr\'ak~\cite{LeeSV09}, using an algorithm with run-time exponential in $p$. For $p = 1$ the best result is a
$3.23$-approximation by Vondr\'ak~\cite{Vondrak09}. In contrast, our algorithms have run time only linear in $p$, but our approximation factors are worse by constant factors for the small values of $p$ where previous results exist. We have not tried to optimize our constants, but it seems likely that
matching, or improving on the previous results  for constant $p$ will need more than just
choosing the parameters carefully. We leave such improvements as an open
problem.

\subsection{Submodular Maximization and Secretary Problems in an Economic Context}

Submodular maximization and secretary problems have both been widely
studied in their economic contexts. The problem of selecting a subset of
people in a social network to maximize their influence in a viral
marketing campaign can be modeled as a constrained submodular
maximization problem~\cite{KKT03,MR07}. When costs are introduced, the
influence minus the cost gives us \emph{non-monotone} submodular
maximization problems; prior to this work, \emph{online} algorithms for
non-monotone submodular maximization problems were not known. Asadpour
et al.\ studied the problem of adaptive stochastic (monotone) submodular
maximization with applications to budgeting and sensor placement
\cite{ANS08}, and Agrawal et al.\ showed that the \emph{correlation gap}
of submodular functions was bounded by a constant using an elegant
cost-sharing argument, and related this result to social welfare
maximizing auctions \cite{ADSY09}. Finally, secretary problems, in which
elements arriving in random order must be selected so as to maximize
some constrained objective function have well-known connections to
online auctions~\cite{Kleinberg-multiple, BIK07, BIKK07,
  HajiaghayiKP04}. Our simpler \emph{offline} algorithms allow us to
generalize these results to give the first secretary algorithms capable
of handling a non-monotone submodular objective function.

\subsection{Our Main Ideas}

At a high level, the simple yet crucial observation for the offline
results is this: many of the previous algorithms and proofs for
constrained monotone submodular maximization can be adapted to show that
the set $S$ produced by them satisfies $f(S) \geq \beta f(S \cup C^*)$,
for some $0 < \beta \leq 1$, and $C^*$ being an optimal solution.  In
the monotone case, the right hand side is at least $f(C^*) = \OPT$ and
we are done.  In the non-monotone case, we cannot do this. However, we
observe that if $f(S \cap C^*)$ is a reasonable fraction of $\OPT$, then
(approximately) finding the most valuable set within $S$ would give us a
large value---and since we work with constraints that are downwards
closed, finding such a set is just \emph{unconstrained} maximization on
$f(\cdot)$ restricted to $S$, for which Feige et al.~\cite{FMV07} give
good algorithms!  On the other hand, if $f(S \cap C^*) \leq \epsilon
\OPT$ and $f(S)$ is also too small, then one can show that deleting the
elements in $S$ and running the procedure again to find another set $S'
\sse \Omega \setminus S$ with $f(S') \geq \beta f(S' \cap (C^* \setminus
S))$ would guarantee a good solution! Details for the specific problems
appear in the following sections; we first consider the simplest
cardinality constraint case in \lref[Section]{sec:submax-card} to
illustrate the general idea, and then give more general results in
\lref[Sections]{sec:submod-indep-sys} and~\ref{sec:subm-knapsack}.

For the secretary case where the elements arrive in random order,
algorithms were not known for the monotone case either---the main
complication being that we cannot run a greedy algorithm (since the
elements are arriving randomly), and moreover the value of an incoming
element depends on the previously chosen set of elements. Furthermore, to
extend the results to the non-monotone case, one needs to avoid the
local-search algorithms (which, in fact, motivated the above results), since these algorithms necessarily implement multiple passes over the input, while the secretary model only allows a single pass over it. The details on all
these are given in \lref[Section]{sec:secretaries}.

\subsection{Related Work}
\label{sec:related-work}

\emph{Monotone Submodular Maximization.}  The (offline) monotone
submodular optimization problem has been long studied: Fisher,
Nemhauser, and Wolsey~\cite{FNW,FNWII} showed that the greedy and
local-search algorithms give a $(\rme/\rme-1)$-approximation with
cardinality constraints, and a $(p+1)$-approximation under $p$ matroid
constraints. In another line of
work,~\cite{Jenkyns76,Korte-Hausmann,HKJ80} showed that the greedy
algorithm is a $p$-approximation for maximizing a \emph{modular} (i.e.,
additive) function subject to a $p$-independence system. This proof
extends to show a $(p+1)$-approximation for monotone submodular
functions under the same constraints (see, e.g.,~\cite{CCPV-journal}).
A long standing open problem was to improve on these results; nothing
better than a $2$-approximation was known even for monotone maximization
subject to a single partition matroid constraint. Calinescu et
al.~\cite{CCPV07} showed how to maximize monotone submodular functions
representable as weighted matroid rank functions subject to any matroid
with an approximation ratio of $(\rme/\rme-1)$, and soon thereafter,
Vondr\'{a}k extended this result to \emph{all} submodular
functions~\cite{Vondrak08}; these highly influential results appear
jointly in~\cite{CCPV-journal}. Subsequently, Lee et al.~\cite{LeeSV09}
give algorithms that beat the $(p+1)$-bound for $p$ matroid constraints
with $p \geq 2$ to get a $(\frac{p^2}{p-1} + \epsilon)$-approximation.

\medskip \emph{Knapsack constraints.} Sviridenko~\cite{Sviridenko04}
extended results of Wolsey~\cite{Wolsey82} and Khuller et
al.~\cite{Khuller-Moss-Naor} to show that a greedy-like algorithm with
partial enumeration gives an $(\rme/\rme-1)$-approximation to monotone
submodular maximization subject to a knapsack constraint. Kulik et
al.~\cite{KST09} showed that one could get essentially the same
approximation subject to a constant number of knapsack constraints.  Lee
et al.~\cite{LMNS-journal} give a $5$-approximation for the same problem
in the non-monotone case.

\medskip \emph{Mixed Matroid-Knapsack Constraints.} Chekuri et
al.~\cite{CVZ10} give strong concentration results for dependent
randomized rounding with many applications; one of these applications is
a $((\rme/\rme - 1) - \epsilon)$-approximation for monotone maximization
with respect to a matroid and any constant number of knapsack
constraints.  \cite[Section~F.1]{GNR-robust} extends ideas
from~\cite{CK05} to give polynomial-time algorithms with respect to
non-monotone submodular maximization with respect to a $p$-system and
$q$ knapsacks: these algorithms achieve an $p+q+O(1)$-approximation for
constant $q$ (since the running time is $n^{\text{poly}(q)}$), or a
$(p+2)(q+1)$-approximation for arbitrary $q$; at a high level, their
idea is to ``emulate'' a knapsack constraint by a polynomial number of
partition matroid constraints.

\medskip
\emph{Non-Monotone Submodular Maximization.}  In the non-monotone case,
even the unconstrained problem is NP-hard (it captures max-cut). Feige,
Mirrokni and Vondr\'{a}k~\cite{FMV07} first gave constant-factor
approximations for this problem. Lee et al.~\cite{LMNS-journal} gave the first
approximation algorithms for constrained non-monotone maximization
(subject to $p$ matroid constraints, or $p$ knapsack constraints); the
approximation factors were improved by Lee et al.~\cite{LeeSV09}. The
algorithms in the previous two papers are based on local-search with
$p$-swaps and would take $n^{\Theta(p)}$ time. Recent work by
Vondr\'{a}k~\cite{Vondrak09} gives much further insight into the
approximability of submodular maximization problems.

\medskip
\emph{Secretary Problems.}  The original secretary problem seeks to
maximize the probability of picking the element in a collection having
the highest value, given that the elements are examined in random
order~\cite{Dynkin63,Freeman-survey,Ferguson-survey}. The problem was
used to model item-pricing problems by Hajiaghayi et
al.~\cite{HajiaghayiKP04}.  Kleinberg~\cite{Kleinberg-multiple} showed
that the problem of maximizing a \emph{modular} function subject to a
cardinality constraint in the secretary setting admits a $(1 +
\frac{\Theta(1)}{\sqrt{k}})$-approximation, where $k$ is the
cardinality.  (We show that maximizing a \emph{submodular} function
subject to a cardinality constraint cannot be approximated to better
than some universal constant, independent of the value of $k$.) Babaioff
et al.~\cite{BIK07} wanted to maximize modular functions subject to
matroid constraints, again in a secretary-setting, and gave
constant-factor approximations for some special matroids, and an $O(\log
k)$ approximation for general matroids having rank $k$.  This line of
research has seen several developments recently~\cite{BIKK07,DimitrovP08,KP08,BDGIT09}.

\subsubsection{Independent Work on Submodular Secretaries}
\label{sec:related-indep}

Concurrently and independently of our work, Bobby Kleinberg has given an
algorithm similar to that in \S\ref{sec:secy-card} for monotone
secretary submodular maximization under a cardinality
constraint~\cite{Bobby}. Again independently, Bateni et al.\ consider
the problem of non-monotone submodular maximization in the secretary
setting \cite{BHZ10}; they give a different $O(1)$-approximation subject
to a cardinality constraint, an $O(L \log^2 k)$-approximation subject to
$L$ matroid constraints, and an $O(L)$-approximation subject to $L$
knapsack constraints in the secretary setting. While we do not consider
multiple constraints, it is easy to extend our results to obtain
$O(L \log k)$ and $O(L)$ respectively using standard techniques.


\subsection{Preliminaries}
\label{sec:preliminaries}

Given a set $S$ and an element $e$, we use $S + e$ to denote $S \cup
\{e\}$. A function $f:2^\Omega \rightarrow\bR_+$ is \emph{submodular} if
for all $S,T \subseteq \Omega$, $f(S) + f(T) \geq f(S \cup T) + f(S \cap
T)$. Equivalently, $f$ is submodular if it has \emph{decreasing marginal
  utility}: i.e., for all $S \subseteq T \subseteq \Omega$, and for all
$e \in \Omega$, $f(S + e) - f(S) \geq f(T + e) - f(T)$. Also, $f$ is
called \emph{monotone} if $f(S) \leq f(T)$ for $S \sse T$. Given $f$ and
$S \subseteq \Omega$, define $f_S: 2^\Omega \to \bR$ as $f_S(A) := f(S
\cup A) - f(S)$. The following facts are standard.

\begin{proposition}
  \label{prob:submodfacts}
  If $f$ is submodular with $f(\emptyset) = 0$, then
  \begin{OneLiners}
  \item for any $S$, $f_S$ is submodular with $f_S(\emptyset) = 0$, and
  \item $f$ is also \emph{subadditive}; i.e., for disjoint sets $A, B$,
    we have $f(A) + f(B) \geq f(A \cup B)$.
  \end{OneLiners}
\end{proposition}


\noindent
\textbf{Matroids.} A \emph{matroid} is a pair $\mathcal{M} = (\Omega, \I
\sse 2^\Omega)$, where $\I$ contains $\emptyset$, if $A \in \I$ and $B
\sse A$ then $B \in \I$, and for every $A, B \in \I$ with $|A| < |B|$,
there exists $e \in B \setminus A$ such that $A + e \in \I$. The sets in
$\I$ are called \emph{independent}, and the \emph{rank} of a matroid is
the size of any maximal independent set (base) in $\mathcal{M}$. In a
\emph{uniform} matroid, $\I$ contains all subsets of size at most $k$. A
\emph{partition} matroid, we have groups $g_1, g_2, \ldots, g_k \sse
\Omega$ with $g_i \cap g_j = \emptyset$ and $\cup_j g_j = \Omega$; the
independent sets are $S \sse \Omega$ such that $|S \cap g_i| \leq 1$.

\medskip
\noindent
\textbf{Unconstrained (Non-Monotone) Submodular Maximization.}  We use
$\FMV_\alpha(S)$ to denote an approximation algorithm given by Feige,
Mirrokni, and Vondr\'{a}k~\cite{FMV07} for unconstrained submodular
maximization in the non-monotone setting: it returns a set $T \sse S$
such that $f(T) \geq \frac1\alpha \max_{T' \sse S} f(T')$. In fact,
Feige et al.\ present many such algorithms, the best approximation ratio
among these is $\alpha = 2.5$ via a local-search algorithm, the easiest
is a $4$-approximation that just returns a uniformly random subset of
$S$.

\section{Submodular Maximization subject to a Cardinality Constraint}
\label{sec:submax-card}

We first give an offline algorithm for submodular maximization subject
to a cardinality constraint: this illustrates our simple approach, upon
which we build in the following sections. Formally, given a subset $X
\sse \Omega$ and a non-negative submodular function $f$ that is
potentially non-monotone, but has $f(\emptyset) = 0$. We want to
approximate $\max_{S \sse X: |S| \leq k} f(S)$. The greedy algorithm
starts with $S \gets \emptyset$, and repeatedly picks an element $e$
with maximum marginal value $f_S(e)$ until it has $k$ elements.

\begin{lemma}
  \label{lem:det-threshold}
  For any set $|C| \leq k$, the greedy algorithm returns a set $S$ that
  satisfies $f(S) \geq \frac12 \, f(S \cup C)$.
\end{lemma}

\begin{proof}
  Suppose not. Then $f_S(C) = f(S \cup C) - f(S) > f(S)$, and hence
  there is at least one element $e \in C \setminus S$ that has
  $f_S(\{e\}) > \frac{f(S)}{|C \setminus S|} > \frac{f(S)}{k}$. Since we
  ran the greedy algorithm, at each step this element $e$ would have
  been a contender to be added, and by submodularity, $e$'s marginal
  value would have been only higher then. Hence the elements actually
  added in each of the $k$ steps would have had marginal value more than
  $e$'s marginal value at that time, which is more than $f(S)/k$. This
  implies that $f(S) > k \cdot f(S)/k$, a contradiction.
\end{proof}
This theorem is existentially tight: observe that if the function $f$ is
just the cardinality function $f(S) = |S|$, and if $S$ and $C$ happen to
be disjoint, then $f(S) = \frac12 f(S \cup C)$.

\begin{lemma}[Special Case of Claim~2.7 in~\cite{LMNS-journal}]
  \label{lem:cross-sums}
  Given sets $C, S_1 \subseteq U$, let $C' = C
  \setminus S_1$, and $S_2 \subseteq U \setminus S_1$. Then
   $ f(S_1 \cup C) + f(S_1 \cap C) + f(S_2 \cup C')
     \geq f(C)$.
\end{lemma}

\begin{proof}
  By submodularity, it follows that $f(S_1 \cup C) + f(S_2 \cup C') \geq
  f(S_1 \cup S_2 \cup C) + f(C')$. Again using submodularity, we get $
  f(C') + f(S_1 \cap C) \geq f(C) + f(\emptyset)$.  Putting these
  together and using non-negativity of $f(\cdot)$, the lemma follows.
\end{proof}

\begin{wrapfigure}{r}{0.45\textwidth}
  \hrule\medskip
  \begin{algorithmic}[1]
    \STATE \textbf{let} $X_1 \gets X$

    \FOR{$i = 1$ to $2$}

    \STATE \textbf{let} $S_i \gets$ Greedy$(X_i)$
    \label{step:greedy}

    \STATE \textbf{let} $S'_i \gets$ \FMV$_\alpha(S_i)$
    \label{step:fmv}

    \STATE \textbf{let} $X_{i+1} \gets X_i \setminus S_i$.

    \ENDFOR

    \RETURN best of $S_1, S_1', S_2$. \label{step:ret}
  \end{algorithmic}
  \medskip\hrule
  \caption{Submod-Max-Cardinality$(X,k,f)$}
  \label{alg:card}
  \vspace{-20pt}
\end{wrapfigure}

We now give our algorithm Submod-Max-Cardinality
(\lref[Figure]{alg:card}) for submodular
maximization: it has the same multi-pass structure as that of Lee et
al., but uses the greedy analysis above instead of a local-search
algorithm.

\begin{theorem}
  \label{thm:onepass}
  The algorithm Submod-Max-Cardinality is a $(4+\alpha)$-approximation.
\end{theorem}

\begin{proof}
  Let $C^*$ be the optimal solution with $f(C^*) = \OPT$.  We know that
  $f(S_1) \geq \frac12 f(S_1 \cup C^*)$. Also, if $f(S_1 \cap C^*)$ is at least
  $\epsilon\, \OPT$, then we know that the $\alpha$-approximate
  algorithm \FMV$_{\alpha}$ gives us a value of at least
  $(\epsilon/\alpha)\OPT$. Else,
  \begin{gather}
    \label{eq:1}
    \tsty f(S_1) \geq \frac12 f(S_1 \cup C^*)  \geq \frac12 f(S_1 \cup C^*) +
    \frac12 f(S_1 \cap
    C^*) - \epsilon\, \OPT/2
  \end{gather}
  Similarly, we get that $f(S_2) \geq \frac12 f(S_2 \cup (C^* \setminus
  S_1))$. Adding this to~(\ref{eq:1}), we get
  \begin{align}
    2 \max(f(S_1), f(S_2)) &\geq f(S_1) + f(S_2) \notag \\
    &\geq \tsty \frac12 \big( f(S_1 \cup C^*) + f(S_1 \cap C^*) + f(S_2
    \cup (C^* \setminus S_1)) \big) - \epsilon \OPT/2 \label{eq:2} \\
    & \geq \tsty \frac12 f(C^*) - \epsilon \OPT/2 \label{eq:3} \\
    & \geq \tsty \frac12  (1 - \epsilon)\, \OPT \notag.
  \end{align}
  where we used \lref[Lemma]{lem:cross-sums} to get from~(\ref{eq:2})
  to~(\ref{eq:3}). Hence $\max\{f(S_1), f(S_2)\} \geq \frac{1 -
    \epsilon}{4}\, \OPT$. The approximation factor now is $\max\{
  \alpha/\epsilon, 4/(1-\epsilon)\}$. Setting $\epsilon =
  \frac{\alpha}{\alpha + 4}$, we get a $(4 + \alpha)$-approximation, as
  claimed.
\end{proof}

Using the known value of $\alpha = 2.5$ from Feige et al.~\cite{FMV07},
we get a $6.5$-approximation for submodular maximization under
cardinality constraints. While this is weaker than the
$3.23$-approximation of Vondr\'ak~\cite{Vondrak09}, or even the
$4$-approximation we could get from Lee et al.~\cite{LMNS-journal} for this
special case, the algorithm is faster, and the idea behind the
improvement works in several other contexts, as we show in the following
sections.

\section{Fast Algorithms for $p$-Systems and Knapsacks}

In this section, we show our greedy-style algorithms which achieve an
$O(p)$-approximation for submodular maximization over $p$-systems, and a
constant-factor approximation for submodular maximization over a
knapsack. Due to space constraints, many proofs are deferred to the
appendices.

\subsection{Submodular Maximization for Independence Systems}
\label{sec:submod-indep-sys}

Let $\Omega$ be a universe of elements and consider a collection ${\cal
  I} \subseteq 2^\Omega$ of subsets of $\Omega$. $(\Omega,{\cal I})$ is
called an {\em independence system} if (a) $\emptyset \in {\cal I}$, and
(b) if $X \in {\cal I}$ and $Y \subseteq X$, then $Y \in {\cal I}$ as
well. The subsets in ${\cal I}$ are called {\em independent}; for any
set $S$ of elements, an inclusion-wise maximal independent set $T$ of
$S$ is called a {\em basis} of $S$.  For brevity, we say that $T$ is a
basis, if it is a basis of $\Omega$.

\begin{definition}
  Given an independence system $(\Omega,{\cal I})$ and a subset $S
  \subseteq \Omega$. The {\em rank} $r(S)$ is defined as the cardinality
  of the \emph{largest} basis of $S$, and the {\em lower rank} $\rho(S)$
  is the cardinality of the \emph{smallest} basis of $S$. The
  independence system is called a $p$-independence system (or a
  $p$-system) if
  $ \max_{S \subseteq \Omega} \frac{r(S)}{\rho(S)} \leq p $.
\end{definition}
See, e.g.,~\cite{CCPV-journal} for a discussion of independence systems
and their relationship to other families of constraints; it is useful to
recall that intersections of $p$ matroids form a $p$-independent
system.

\subsubsection{The Algorithm for $p$-Independence Systems}

Suppose we are given an independence system $(\Omega, \mathcal{I})$, a
subset $X \sse \Omega$ and a non-negative submodular function $f$ that
is potentially non-monotone, but has $f(\emptyset) = 0$. We want to find
(or at least approximate) $\max_{S \sse X: S \in \mathcal{I}} f(S)$. The
greedy algorithm for this problem is what you would expect: start with
the set $S = \emptyset$, and at each step pick an element $e \in X
\setminus S$ that maximizes $f_S(e)$ and ensures that $S + e$ is
also independent.  If no such element exists, the algorithm terminates,
else we set $S \gets S + e$, and repeat. (Ideally, we would also
check to see if $f_S(e) \leq 0$, and terminate at the first time this
happens; we don't do that, and instead we add elements even when the
marginal gain is negative until we cannot add any more elements without
violating independence.) The proof of the following lemma appears in
\lref[Section]{sec:psystem-app}, and closely follows that for the monotone
case from~\cite{CCPV-journal}.

\begin{lemma}
  \label{lem:approx}
  For a $p$-independence system, if $S$ is the independent set returned
  by the greedy algorithm, then for any independent set $C$, $f(S) \geq
  \frac{1}{p+1} f(C \cup S)$.
\end{lemma}

\begin{wrapfigure}{r}{0.45\textwidth}
    \hrule\medskip
    \begin{algorithmic}[1]
      \STATE $X_1 \gets X$

      \FOR{$i = 1$ to $p+1$}

      \STATE \label{step:greed} $S_i \gets$ Greedy$(X_i, \I, f)$

      \STATE $S_i' \gets \FMV_\alpha(S_i)$

      \STATE $X_{i+1} \gets X_i \setminus S_i$

      \ENDFOR

      \RETURN $S \gets$ best among $\{S_i\}_{i = 1}^{p+1} \cup \{S_i'\}_{i
      = 1}^{p+1}$.
  \end{algorithmic}
  \medskip\hrule
  \caption{Submod-Max-$p$-System$(X, \I, f)$}
  \vspace{-10pt}
  \label{alg:psys}
\end{wrapfigure}

The algorithm Submod-Max-$p$-Systems (\lref[Figure]{alg:psys}) for
maximizing a non-monotone submodular function $f$ with $f(\emptyset) =
0$ over a $p$-independence system now immediately suggests itself.

\begin{theorem}
  \label{thm:greed}
  The algorithm Submod-Max-$p$-System is a $(1+\alpha) (p + 2 +
  1/p)$-approximation for maximizing a non-monotone submodular function
  over a $p$-independence system, where $\alpha$ is the approximation
  guarantee for unconstrained (non-monotone) submodular maximization.
\end{theorem}

\begin{proof}
  Let $C^*$ be an optimal solution with $\OPT = f(C^*)$, and let $C_i = C^*
  \cap X_i$ for all $i \in [p+1]$---hence $C_1 = C^*$. Note that $C_i$ is
  a feasible solution to the greedy optimization in
  \lref[Step]{step:greed}. Hence, by \lref[Lemma]{lem:approx}, we know that
  $f(S_i) \geq \frac{1}{p+1}f(C_i \cup S_i)$. Now, if for some $i$, it
  holds that $f(S_i \cap C_i) \geq \epsilon \OPT$ (for $\epsilon > 0$ to
  be chosen later), then the guarantees of $\FMV_\alpha$ ensure that
  $f(S_i') \geq (\epsilon \OPT)/\alpha$, and we will get a
  $\alpha/\epsilon$-approximation. Else, it holds for all $i \in [p+1]$
  that
  \begin{gather}
    \label{eq:15}
    \tsty f(S_i) \geq \frac{1}{p+1}f(C_i \cup S_i) + f(C_i \cap S_i) -
    \epsilon\,\OPT
  \end{gather}
  Now we can add all these inequalities, divide by $p+1$, and use the
  argument from~\cite[Claim~2.7]{LMNS-journal} to infer that
  \begin{gather}
    \tsty f(S) \geq \frac{p}{(p+1)^2} f(C^*) - \epsilon\,\OPT = \OPT
    \left(\frac{p}{(p+1)^2} - \epsilon \right).
  \end{gather}
  (While Claim~2.7 of~\cite{LMNS-journal} is used in the context of a
  local-search algorithm, it uses just the submodularity of the function
  $f$, and the facts that $(\cup_{j < i}S_j \cup C) \cap (S_i \cup C_i)
  = C_i$ and $(\cup_{j < i} (S_j \cap C_j) \cup C_i = C$ for every $i$.)
  Thus the approximation factor is $\max\{\alpha/\epsilon,
  (\frac{p}{(p+1)^2} - \epsilon)^{-1}\}$.  Setting $\epsilon =
  \frac{\alpha}{1+\alpha}\frac{p}{(p+1)^2}$, we get the claimed
  approximation ratio.
\end{proof}

Note that even using $\alpha = 1$, our approximation factors differ from the
ratios in Lee et al.~\cite{LMNS-journal,LeeSV09} by a small constant factor. However, the proof here is
somewhat simpler and also works seamlessly for all $p$-independence
systems instead of just intersections of matroids. Moreover our running
time is only linear in the number of matroids, instead of being
exponential as in the local-search: previously, no polynomial time algorithms were known for this problem if $p$ was super-constant. Note that running the algorithm just
twice instead of $p+1$ times reduces the run-time further; we can then
use \lref[Lemma]{lem:cross-sums} instead of the full power of
\cite[Claim~2.7]{LMNS-journal}, and hence the constants are slightly worse.



\subsection{Submodular Maximization over Knapsacks}
\label{sec:subm-knapsack}

The paper of Sviridenko~\cite{Sviridenko04} gives a greedy algorithm
with partial enumeration that achieves a
$\frac{\rme}{\rme-1}$-approximation for \emph{monotone} submodular
maximization with respect to a knapsack constraint. In particular, each
element $e \in X$ has a size $c_e$, and we are given a bound $B$: the
goal is to maximize $f(S)$ over subsets $S \sse X$ such that $\sum_{e
  \in S} c_e \leq B$. His algorithm is the following---for each possible
subset $S_0 \subseteq X$ of at most three elements, start with $S_0$ and
iteratively include the element which maximizes the gain in the function
value per unit size, and the resulting set still fits in the knapsack.
(If none of the remaining elements gives a positive gain, or fit in the
knapsack, stop.)  Finally, from among these $O(|X|^3)$ solutions, choose
the best one---Sviridenko shows that in the monotone submodular case,
this is an $\frac{\rme}{\rme-1}$-approximation algorithm.  One can modify
Sviridenko's algorithm and proof to show the following result for
non-monotone submodular functions. (The details are in
\lref[Appendix]{sec:sviridenko}).
\begin{theorem}
  \label{thm:sviri}
  There is a polynomial-time algorithm that given the above input,
  outputs a polynomial sized collection of sets such that for any valid
  solution $C$, the collection contains a set $S$ satisfying $f(S) \geq
  \frac12 f(S \cup C)$.
\end{theorem}
Note that the tight example for cardinality constraints shows that we
cannot hope to do better than a factor of $1/2$. Now using an argument
very similar to that in \lref[Theorem]{thm:onepass} gives us the
following result for non-monotone submodular maximization with respect
to a knapsack constraint.
\begin{theorem}
  There is an $(4+\alpha)$-approximation for the problem of maximizing a
  submodular function with respect a knapsack constraint, where $\alpha$
  is the approximation guarantee for unconstrained (non-monotone)
  submodular maximization.
\end{theorem}

\section{Constrained Submodular Maximization in the Secretary Setting}
\label{sec:secretaries}

In this section, we will give algorithms for submodular maximization in
the secretary setting: first subject to a cardinality constraint, then
with respect to a partition matroid, and finally an algorithm for
general matroids. The main algorithmic concerns tackled in this section
when developing secretary algorithms are: (a)~previous algorithms for
non-monotone maximization required local-search, which seems difficult
in an online secretary setting, so we developed greedy-style algorithms;
(b)~we need multiple passes for non-monotone optimization, and while
that can be achieved using randomization and running algorithms in
parallel, these parallel runs of the algorithms may have correlations
that we need to control (or better still, avoid); and of course (c)~the
marginal value function changes over the course of the algorithm's
execution as we pick more elements---in the case of partition matroids,
e.g., this ever-changing function creates several complications.

We also show an information theoretic lower bound: no secretary
algorithm can approximately maximize a submodular function subject to a
cardinality constraint $k$ to a factor better than some universal
constant greater than 1, independent of $k$ (This is ignoring
computational constraints, and so the computational inapproximability of
offline submodular maximization does not apply). This is in contrast to
the additive secretary problem, for which Kleinberg gives a secretary
algorithm achieving a $\smash{\frac{1}{1-5/\sqrt{k}}}$-approximation
\cite{Kleinberg-multiple}. This lower bound is found in
\lref[Appendix]{sec:lower-bounds}. (For a discussion about independent
work on submodular secretary problems, see \S\ref{sec:related-indep}.)

\subsection{Subject to a Cardinality Constraint}
\label{sec:secy-card}

The offline algorithm presented in \lref[Section]{sec:submax-card}
builds three potential solutions and chooses the best amongst them. We
now want to build just one solution in an \emph{online} fashion, so that
elements arrive in random order, and when an element is added to the
solution, it is never discarded subsequently. We first give an online
algorithm that is given the optimal value $\OPT$ as input but where the
elements can come in \emph{worst-case} order (we call this an ``online
algorithm with advice''). Using sampling ideas we can estimate $\OPT$,
and hence use this advice-taking online algorithm in the secretary model
where elements arrive in random order.

To get the advice-taking online algorithm, we make two changes. First,
we do not use the greedy algorithm which selects elements of highest
marginal utility, but instead use a \emph{threshold algorithm}, which
selects any element that has marginal utility above a certain threshold.
Second, we will change \lref[Step]{step:fmv} of Algorithm
Submod-Max-Cardinality to use FMV$_4$, which simply selects a random
subset of the elements to get a $4$-approximation to the unconstrained
submodular maximization problem~\cite{FMV07}. The \emph{Threshold
  Algorithm} with inputs $(\tau,k)$ simply selects each element as it
appears if it has marginal utility at least $\tau$, up to a maximum of
$k$ elements.
\begin{lemma}[Threshold Algorithm]
  \label{lem:threshold-semion}
  Let $C^*$ satisfy $f(C^*) = \OPT$. The threshold algorithm on inputs
  $(\tau,k)$ returns a set $S$ that either has $k$ elements and hence a
  value of at least $\tau k$, or a set $S$ with value $f(S) \geq f(S
  \cup C^*) - |C^*|\tau$.
\end{lemma}

\begin{proof}
  The claim is immediate if the algorithm picks $k$ elements, so suppose
  it does not pick $k$ elements, and also $f(S) < f(S \cup C^*) - |C^*|\tau$.
  Then $f_S(C^*) > |C^*|\tau$, or $\tau < \frac{f_S(C^*)}{|C^*|} \leq \frac{\sum_{e
      \in C^*} f_S(e)}{|C^*|}$. By averaging, this implies there exists an
  element $e \in C^*$ such that $f_S(e) > \tau$; this element cannot have
  been chosen into $S$ (otherwise the marginal value would be $0$), but
  it would have been chosen into $S$ when it was considered by the
  algorithm (since at that time its marginal value would only have been
  higher). This gives the desired contradiction.
\end{proof}

\begin{theorem}
  \label{thm:rand-return}
  If we change Algorithm Submod-Max-Cardinality from
  \S\ref{sec:submax-card} to use the threshold algorithm with threshold
  $\tau = \frac{\OPT}{7k}$ in \lref[Step]{step:greedy}, and to use the
  random sampling algorithm \FMV$_4$ in \lref[Step]{step:fmv}, and
  return a (uniformly) random one of $S_1, S_1', S_2$ in
  \lref[Step]{step:ret}, the expected value of the returned set is at
  least $\OPT/21$.
\end{theorem}

\begin{proof}
  We show that $f(S_1) + f(S_1') + f(S_2) \geq \tau k = \frac{\OPT}{7}$,
  and picking a random one of these gets a third of that in expectation.
  Indeed, if $S_1$ or $S_2$ has $k$ elements, then $f(S_1) + f(S_2) \geq
  \tau k$. Else if $f(S_1 \cap C^*) \geq 4 \tau k$, then \FMV$_4$
  guarantees that $f(S_1') \geq \tau k$. Else $f(S_1) + f(S_2) \geq
  (f(S_1 \cup C^*) - \tau k) + (f(S_2 \cup C^*) - \tau k) + (f(S_1 \cap
  C^*) - 4\tau k)$, which by \lref[Lemma]{lem:cross-sums} is at least
  $\OPT - 6 \tau k = \tau k$.
\end{proof}

\begin{observation}
  Given the value of $\OPT$, the algorithm of
  \lref[Theorem]{thm:rand-return} can be implemented in an online
  fashion where we (irrevocably) pick at most $k$ elements.
\end{observation}

\begin{proof}
  We can randomly choose which one of $S_1, S_1', S_2$ we want to output
  before observing any elements. Clearly $S_1$ can be determined online,
  as can $S_2$ by choosing any element that has high marginal value and
  is not chosen in $S_1$. Moreover, $S_1'$ just selects elements from
  $S_1$ independently with probability $1/2$.
\end{proof}

\begin{observation}
  In both the algorithms of \lref[Theorems]{thm:onepass}
  and~\ref{thm:rand-return}, if we use some value $Z \leq \OPT$ instead
  of $\OPT$, the returned set has value at least $Z/(4+\alpha)$, and
  expected value at least $Z/21$, respectively.
\end{observation}

Finally, it will be convenient to recall Dynkin's algorithm: given a
stream of $n$ numbers randomly ordered, it samples the first $1/e$ fraction
of the numbers and picks the next element that is larger than all
elements in the sample.

\subsubsection{The  Secretary Algorithm for the Cardinality Case}

\begin{wrapfigure}{r}{0.6\textwidth}
  \centering
  \vspace{-10pt}
  \hrule\medskip
  \begin{algorithmic}
    \STATE \textbf{Let} Solution $\leftarrow \emptyset$.

    \STATE \textbf{Flip} a fair coin

    \IF {heads}

    \STATE Solution $\leftarrow$ most valuable item using Dynkin's-Algo

    \ELSE

    \STATE \textbf{Let} $m \in B(n,1/2)$ be a draw from the binomial
    distribution

    \STATE $A_1 \gets$ $\rho_{\text{off}}$-approximate offline algorithm
    on the first $m$ elements.

    \STATE $A_2 \gets$  $\rho_{\text{on}}$-approximate advice-taking online
    algorithm with \\ \qquad \qquad $f(A_1)$ as the guess for $\OPT$.

    \STATE Return $A_2$
    \ENDIF
  \end{algorithmic}
  \medskip \hrule
  \caption{\textbf{Algorithm} SubmodularSecretaries}
  \label{theSecyAlgorithm}
  \vspace{-10pt}
\end{wrapfigure}

For a constrained submodular optimization, if we are given \emph{(a)} a
$\rho_{\text{off}}$-approximate offline algorithm, and also \emph{(b)} a
$\rho_{\text{on}}$-approximate online advice-taking algorithm that works
given an estimate of $\OPT$, we can now get an algorithm in the
secretary model thus: we use the offline algorithm to estimate $\OPT$ on
the first half of the elements, and then run the advice-taking online
algorithm with that estimate. The formal algorithm appears in
\lref[Figure]{theSecyAlgorithm}.  Because of space constraints, we have
deferred the proof of the following theorem to
\lref[Appendix]{sec:secretariesproofs}.
\begin{theorem}
  \label{thm:conversion}
  The above algorithm is an $O(1)$-approximation algorithm for the
  cardinality-constrained submodular maximization problem in the
  secretary setting.
\end{theorem}


\subsection{Subject to a Partition Matroid Constraint}
\label{sec:secy-part}

In this section, we give a constant-factor approximation for maximizing
submodular functions subject to a partition matroid. Recall that in such
a matroid, the universe is partitioned into $k$ ``groups'', and the
independent sets are those which contain at most one element from each
group. To get a secretary-style algorithm for \emph{modular (additive)}
function maximization subject to a partition matroid, we can run
Dynkin's algorithm on each group independently.  However, if we have a
submodular function, the marginal value of an element depends on the
elements previously picked---and hence the marginal value of an element
as seen by the online algorithm and the adversary become very different.

We first build some intuition by considering a simpler ``contiguous
partitions'' model where all the elements of each group arrive together
(in random order), but the groups of the partition are presented in some
\emph{arbitrary} order $g_1, g_2, \ldots, g_r$. We then go on to handle
the case when all the elements indeed come in completely random order,
using what is morally a reduction to the contiguous partitions case.

\subsubsection{A Special Case: Contiguous Partitions}
\label{sec:contiguous}

For the contiguous case, one can show that executing Dynkin's algorithm
with the obvious marginal valuation function is a good algorithm: this
is not immediate, since the valuation function changes as we pick some
elements---but it works out, since the groups come contiguously.  Now,
as in the previous section, one wants to run two parallel copies of this
algorithm (with the second one picking elements from among those not
picked by the first)---but the correlation causes the second algorithm
to not see a random permutation any more! We get around this by coupling
the two together as follows:
\begin{quote}
  Initially, the algorithm determines whether it is one of 3 different
  modes (A, B, or C) uniformly at random. The algorithm maintains a set
  of selected elements, initially $S_0$. When group $g_i$ of the
  partition arrives, it runs Dynkin's secretary algorithm on the
  elements from this group using valuation function $f_{S_{i-1}}$. If
  Dynkin's algorithm selects an element $x$, our algorithm flips a coin.
  If we are in modes $A$ or $B$, we let $S_i \leftarrow S_{i-1} \cup
  \{x\}$ if the coin is heads, and let $S_i \leftarrow S_{i-1}$
  otherwise. If we are in mode $C$, we do the reverse, and let $S_i
  \leftarrow S_{i-1} \cup \{x\}$ if the coin is tails, and let $S_i
  \leftarrow S_{i-1}$ otherwise. Finally, after the algorithm has
  completed, if we are in mode $B$, we discard each element of $S_r$
  with probability $1/2$. (Note that we can actually implement this step
  online, by 'marking' but not selecting elements with probability $1/2$
  when they arrive).
\end{quote}

\begin{lemma}
  \label{lem:partn-lem1}
  The above algorithm is a $(3+6\mathrm{e})$-approximation for the
  submodular maximization problem under partition matroids, when each
  group of the partition comes as a contiguous segment.
\end{lemma}

\begin{proof}
  We first analyze the case in which the algorithm is in mode $A$ or
  $C$.  Consider a hypothetical run of \emph{two} versions of our
  algorithm simultaneously, one in mode $A$ and one in mode $C$ which
  share coins and produce sets $S^A_r$ and $S^C_r$. The two algorithms
  run with identical marginal distributions, but are coupled such that
  whenever both algorithms attempt to select the same element (each with
  probability $1/2$), we flip only one coin, so one succeeds while the
  other fails.  Note that $S^C_r \subseteq U\setminus S^A_r$, and so we
  will be able to apply \lref[Lemma]{lem:cross-sums}.  For a fixed
  permutation $\pi$, let $S^A_r(\pi)$ be the set chosen by the mode $A$
  algorithm for that particular permutation.  As usual, we
  define $f_A(B) = f(A\cup B) - f(A)$. Hence, $f(S^A_r(\pi)) =
  f(S^A_r(\pi) \cup C^*) - f_{S^A_r(\pi)}(C^*)$, and taking
  expectations, we get
  \begin{align}
    \E[f(S^A_r)] &= \E[f(S^A_r \cup C^*)] - \E[f_{S^A_r}(C^*)] \label{eq:10}
  \end{align}
   Now, for any $e \in X$, let $j(e)$ be the index
  of the group containing $e$; hence we have
  \begin{align}
    \E[f_{S^A_r}(C^*)] &\leq \sum_{e \in C^*} \E[f_{S^A_r}(\{e\})] \leq \sum_{e \in C^*} \E[f_{S^A_{j(e)-1}}(\{e\})] \quad \leq \sum_{e \in C^*} 2\mathrm{e} \cdot
    \E[f_{S^A_{j(e)-1}}(\{Y_{j(e)}\})] \notag \\
    &= 2\mathrm{e} \cdot \E[f(S^A_r)], \label{eq:11}
  \end{align}
  where the first inequality is just subadditivity, the second
  submodularity, the third follows from the fact that Dynkin's algorithm
  is an $\mathrm{e}$-approximation for the secretary problem and
  selecting the element that Dynkin's selects with probability $1/2$
  gives a $2e$ approximation, and the resulting telescoping sum gives
  the fourth equality.  Now substituting~(\ref{eq:11})
  into~(\ref{eq:10}) and rearranging, we get $\E[f(S^A_r)] \geq
  \frac{1}{1+2\mathrm{e}}\;f(S^A_r \cup C^*)$.  An identical analysis of
  the second hypothetical algorithm gives: $\E[f(S^C_r)] \geq
  \frac{1}{1+2\mathrm{e}}\;f(S^C_r \cup C^*\setminus S^A_r)$.

  It remains to analyze the case in which the algorithm runs in mode
  $B$. In this case, the algorithm generates a set $S^B_r$ by selecting
  each element in $S^A_r$ uniformly at random. By the theorem of
  \cite{FMV07}, uniform random sampling achieves a $4$-approximation to
  the problem of \emph{unconstrained} submodular maximization.
  Therefore, we have in this case: $\E[f(S^B_r)] \geq \frac{1}{4}f(S^A_r
  \cap C^*)$.  By \lref[Lemma]{lem:cross-sums}, we therefore have:
  $\E[f(S^A_r)]+\E[f(S^B_r)]+\E[f(S^C_r)] \geq \frac{1}{1+2e}f(C^*)$.
  Since our algorithm outputs one of these three sets uniformly at
  random, it gets a $(3 + 6e)$ approximation to $f(C^*)$.
\end{proof}

\subsubsection{General Case}
We now consider the general secretary setting, in which the elements come in random order, not necessarily grouped by partition. Our previous approach will not work: we cannot simply run Dynkin's secretary algorithm on contiguous chunks of elements, because some elements may be blocked by our previous choices. We instead do something similar in spirit: we divide the elements up into $k$ `epochs', and attempt to select a single element from each. We treat every element that arrives before the current epoch as part of a sample, and according to the current valuation function at the beginning of an epoch, we select the first element that we encounter that has higher value than any element from its own partition group in the sample, so long as we have not already selected something from the same partition group. Our algorithm is as follows:

\begin{quote}
  Initially, the algorithm determines whether it is one of 3 different
  modes (A, B, or C) uniformly at random. The algorithm maintains a set
  of selected elements, initially $S_0$, and observes the first $N_0 \sim B(n,\frac{1}{2})$ of the elements without selecting anything. The algorithm then considers $k$ \emph{epochs}, where the $i$th epoch is the set of $N_i\sim B(n,\frac{1}{100k})$ contiguous elements after the $(i-1)$th epoch. At epoch $i$, we use valuation function $f_{S_{i-1}}$. 
  If an element has higher value than any element from its own partition group that arrived earlier than epoch $i$, we flip a coin. If we are in modes $A$ or $B$, we let $S_i \leftarrow S_{i-1} \cup \{x\}$ if the coin is heads, and let $S_i \leftarrow S_{i-1}$
  otherwise. If we are in mode $C$, we do the reverse, and let $S_i
  \leftarrow S_{i-1} \cup \{x\}$ if the coin is tails, and let $S_i
  \leftarrow S_{i-1}$ otherwise. After all $k$ epochs have passed, we ignore the remaining elements. Finally, after the algorithm has
  completed, if we are in mode $B$, we discard each element of $S_r$
  with probability $1/2$. (Note that we can actually implement this step
  online, by 'marking' but not selecting elements with probability $1/2$
  when they arrive).
\end{quote}

If we were guaranteed to select an element in every epoch $i$ that was the highest valued element according to $f_{S_{i-1}}$, then the analysis of this algorithm would be identical to the analysis in the contiguous case. This is of course not the case. However, we prove a technical lemma that says that we are ``close enough'' to this case.

\begin{lemma}
\label{lem:partitionTechnicalLemma}
For all partition groups $i$ and epochs $j$, the algorithm selects the highest element from group $i$ (according to the valuation function $f_{S_{j-1}}$ used during epoch $j$) during epoch $j$ with probability at least $\Omega(\frac{1}{k})$.
\end{lemma}

Because of space constraints, we defer the proof of this technical lemma
to \lref[Appendix]{sec:secretariesproofs}.

Note an immediate consequence of the above lemma: if $e$ is the element selected from epoch $j$, by summing over the elements in the optimal set $C^*$  (1 from each of the $k$ partition groups), we get:
$$\E[f_{S_{j-1}}(e)] \geq \Omega(\frac{1}{k})\sum_{e' \in C^*}f_{S_{j-1}}(e') \geq \Omega(\frac{f_{S_{j-1}}(C^*)}{k})$$
Summing over the expected contribution to $S_r$ from each of the $k$
epochs and applying submodularity, we get $\E[f_{S_r^A(C^*)}] \leq
O(\E[f(S^A_r)])$. Using this derivation in place of inequality
\ref{eq:11} in the proof of \lref[Lemma]{lem:partn-lem1} proves that our
algorithm gives an $O(1)$ approximation to the non-monotone submodular
maximization problem subject to a partition matroid constraint.

\subsection{Subject to a General Matroid Constraint}
\label{sec:gen-matroid-secy}

We consider matroid constraints where the matroid is $\mathcal{M} =
(\Omega, \I)$ with rank $k$. Let $w_1 = \max_{e \in \Omega} f(\{e\})$
the maximum value obtained by any single element, and let $e_1$ be the
element that achieves this maximum value. (Note that we do not know
these values up-front in the secretary setting.)  In this section, we
first give an algorithm that gets a set of fairly high value given a
threshold $\tau$. We then show how to choose this threshold, assuming we
know the value $w_1$ of the most valuable element, and why this implies
an advice-taking online algorithm having a logarithmic approximation.
Finally, we show how to implement this in a secretary framework.

\medskip\noindent\textbf{A Threshold Algorithm.}  Given a value $\tau$,
run the following algorithm. Initialize $S_1, S_2 \gets \emptyset$. Go
over the elements of the universe $\Omega$ in \emph{arbitrary} order:
when considering element $e$, add it to $S_1$ if $f_{S_1}(e) \ge
\epsilon \tau$ and $S_1 \cup \{e\}$ is independent, else add it to $S_2$ if $f_{S_2}(e) \ge \epsilon \tau$ and $S_2 \cup \{e\}$ is independent,
else discard it. (We will choose the value of $\epsilon$ later.)
Finally, output a uniformly random one of $S_1$ or $S_2$.

To analyze this algorithm, let $C^*$ be the optimal set with $f(C^*) =
\OPT$. Order the elements of $C^*$ by picking its elements greedily
based on marginal values. Given $\tau > 0$, let $C^*_\tau \sse C^*$ be the
elements whose marginal benefit was at least $\tau$ when added in this
greedy order: note that $f(C^*_\tau) \geq |C^*_\tau|\tau$.

\begin{lemma}
  For $\epsilon = 2/5$, the set produced by our algorithm has expected value is at least
  $|C^*_\tau| \cdot \tau/10$.
\end{lemma}

\begin{proof}
  If either $|S_1|$ or $|S_2|$ is at least $|C^*_\tau|/4$, we get value
  at least $|C^*_\tau|/4 \cdot \epsilon \tau$. Else both these sets have
  small cardinality. Since we are in a matroid, there must be a set $A
  \sse C^*_\tau$ of cardinality $|A| \geq |C^*_\tau| - |S_1| - |S_2|
  \geq |C^*_\tau|/2$, such that $A$ is disjoint from both $S_1$ and
  $S_2$, and both $S_1 \cup A$ and $S_2 \cup A$ lie in $\I$ (i.e., they
  are independent).

  We claim that $f(S_1) \geq f(S_1 \cup A) - |A| \cdot
  \epsilon \tau$.  Indeed, an element in $e \in A$ was not added
  by the threshold algorithm; since it could be added while maintaining
  independence, it must have been discarded because the marginal value
  was less than $\epsilon \tau$. Hence $f_{S_1}(\{e\}) < \epsilon \tau$,
  and hence $f(S_1 \cup A) - f(S_1) = f_{S_1}(A) \leq
  \sum_{e \in A} f_{S_1}(\{e\}) < |A|\cdot \epsilon \tau$.
  Similarly, $f(S_2) \geq f(S_2 \cup A) - |A| \cdot \epsilon
  \tau$. And by disjointness, $f(S_1 \cap A) = f(\emptyset) = 0$.
  Hence, summing these and applying \lref[Lemma]{lem:cross-sums}, we get
  that $f(S_1) + f(S_2) \geq f(S_1 \cup A) + f(S_2 \cup A) +
  f(S_1 \cap A) - 2\epsilon\tau |A| \geq f(A) -
  2\epsilon\tau |A|$.

  Since the marginal values of all the elements in $C^*_\tau$ were at
  least $\tau$ when they were added by the greedy ordering, and $A \sse
  C^*_\tau$, submodularity implies that $f(A) \geq |A| \tau$, which in
  turn implies $f(S_1) + f(S_2) \geq (1 - 2\epsilon)\tau |A| \geq (1 -
  2\epsilon) \tau |C^*_\tau|/2$. A random one of $S_1, S_2$ gets half of
  that in expectation. Taking the minimum of $|C^*_\tau|/4 \cdot
  \epsilon \tau$ and $(1 - 2\epsilon) \tau |C^*_\tau|/2$ and setting
  $\epsilon = 2/5$, we get the claim.
\end{proof}

\begin{lemma}
  \label{lem:lbd}
  $\sum_{i = 0}^{\log 2k} |C^*_{w_1/2^i}| \cdot \frac{w_1}{2^i} \geq
  f(C^*)/4 = \OPT/4$.
\end{lemma}

\begin{proof}
  Consider the greedy enumeration $\{e_1, e_2,\ldots, e_t\}$ of $C$, and
  let $w_j = f_{\{e_1, e_2,\ldots, e_{i-1}\}}(\{e_j\})$.  First consider
  an infinite summation $\sum_{i = 0}^{\infty} |C^*_{w_1/2^i}| \cdot
  \frac{w_1}{2^i}$---each element $e_j$ contributes at least $w_j/2$ to
  it, and hence the summation is at least $\frac12 \sum_j w_j$.  But
  $f(C^*) = \sum_{j=1}^t w_j$, which says the infinite sum is at least
  $f(C^*)/2 = \OPT/2$. But the finite sum merely drops a contribution of
  $w_1/4k$ from at most $|C^*| \leq k$ elements, and clearly $\OPT$ is at
  least $w_1$, so removing this contribution means the finite sum is at
  least $\OPT/4$.
\end{proof}

Hence, if we choose a value $\tau$ uniformly from $w_1, w_1/2, w_1/4,
\ldots, w_1/2k$ and run the above threshold algorithm with that setting
of $\tau$, we get that the expected value of the set output by the
algorithm is:
\begin{gather}
  \tsty \frac{1}{1 + \log 2k} \sum_{i = 0}^{\log 2k} |C^*_{w_1/2^i}|
  \cdot \frac{w_1}{10 \cdot 2^i} \geq \frac{1}{1 + \log 2k}
  \frac{\OPT}{40}.
\end{gather}

\medskip\noindent\textbf{The Secretary Algorithm.}
The secretary algorithm for general matroids is the following:
\begin{quote}
  Sample half the elements, let $W$ be the weight of the highest weight
  element in the first half. Choose a value $i \in \{0, 1, \ldots, 2 +
  \log 2k\}$ uniformly at random. Run the threshold algorithm with $W/2^i$ as the threshold
\end{quote}

\begin{lemma}
  The algorithm is an $O(\log k)$-approximation in the secretary
  setting for rank~$k$ matroids.
\end{lemma}

\begin{proof}
  With probability $\Theta(1/\log k)$, we choose the value $i=0$. In
  this case, with constant probability the element with second-highest
  value comes in the first half, and the highest-value element $e_1$
  comes in the second half; hence our (conditional) expected value in
  this case is at least $w_1$. In case this single element accounts for
  more than half of the optimal value, we get $\Omega(\OPT/\log k)$. We
  ignore the case $i = 1$. If we choose $i \geq 2$, now with constant
  probability $e_1$ comes in the first half, implying that $W = w_1$.
  Moreover, each element in $C - e_1$ appears in the second half with
  probability slightly higher than $1/2$. Since $e_1$ accounts for at
  most half the optimal value, the expected optimal value in the second
  half is at least $\OPT/4$. The above argument  then ensures that we get value $\Omega(\OPT/\log k)$ in expectation.
\end{proof}




\medskip\noindent\textbf{Acknowledgments.} We thank C.\ Chekuri, V.\
Nagarajan, M.I.\ Sviridenko, J.\ Vondr\'{a}k, and especially R.D.\
Kleinberg for valuable comments, suggestions, and conversations. Thanks
to C.\ Chekuri also for pointing out an error in
Section~\ref{sec:sviridenko}, and to M.T.\ Hajiaghayi for informing us
of the results in~\cite{BHZ10}.

{\small
\bibliographystyle{alpha}
\bibliography{SubmodularSecretaries}
}

\appendix

\section{Proof of Main Lemma for $p$-Systems}
\label{sec:psystem-app}

Let $e_1, e_2, \ldots, e_k$ be the elements added to $S$ by greedy, and
let $S_i$ be the first $i$ elements in this order, with $\delta_i =
f_{S_{i-1}}(\{e_i\}) = f(S_i) - f(S_{i-1})$, which may be positive or
negative. Since $f(\emptyset) = 0$, we have $f(S = S_k) = \sum_i
\delta_i$. And since $f$ is submodular, $\delta_i \geq \delta_{i+1}$ for
all $i$. 

\begin{lemma}[following \cite{CCPV-journal}]
  \label{lem:approx1}
  For any independent set $C$, it holds that $f(S_k) \geq \frac{1}{p+1}
  f(C \cup S_k)$.
\end{lemma}

\begin{proof}
  We show the existence of a partition of $C$ into $C_1, C_2, \ldots,
  C_k$ with the following two properties:
  \begin{itemize}
  \item for all $i \in [k]$, $p_1 + p_2 + \ldots + p_i \leq i\cdot p$
    where $p_i := |C_i|$, and
  \item for all $i \in [k]$, $p_i \delta_i \geq f_{S_k}(C_i)$.
  \end{itemize}
  Assuming such a partition, we can complete the proof thus:
  \begin{gather}
    \label{eq:12}
    p \sum_i \delta_i \geq \sum_i p_i \delta_i \geq \sum_i f_{S_k}(C_i)
    \geq f_{S_k}(C) = f(S_k \cup C) - f(S_k),
  \end{gather}
  where the first inequality follows from~\cite[Claim~A.1]{CCPV-journal}
  (using the first property above, and that the $\delta$'s are
  non-increasing), the second from the second property of the partition
  of $C$, the third from subadditivity of $f_{S_k}(\cdot)$ (which is
  implied by the submodularity of $f$ and applications of both facts in
  \lref[Proposition]{prob:submodfacts}), and the fourth from the
  definition of $f_{S_k}(\cdot)$. Using the fact that $\sum_i \delta_i =
  f(S_k)$, and rearranging, we get the lemma.

  Now to prove the existence of such a partition of $C$. Define $A_0,
  A_1, \ldots, A_k$ as follows: $A_i = \{ e \in C \setminus S_i \mid S_i
  + e \in \I\}$. Note that since $C \in \I$, it follows that $A_0 = C$;
  since the independence system is closed under subsets, we have $A_i
  \subseteq A_{i-1}$; and since the greedy algorithm stops only when
  there are no more elements to add, we get $A_k = \emptyset$. Defining
  $C_i = A_{i-1} \setminus A_{i}$ ensures we have a partition $C_1, C_2,
  \ldots, C_k$ of $C$.

  Fix a value $i$. We claim that $S_i$ is a basis (a maximal independent
  set) for $S_i \cup (C_1 \cup C_2 \cup \ldots \cup C_i) = S_i \cup (C
  \setminus A_i)$. Clearly $S_i \in \I$ by construction; moreover, any
  $e \in (C \setminus A_i) \setminus S_i$ was considered but not
  added to $A_i$ because $S_i + e \not\in \I$. Moreover, $(C_1 \cup C_2
  \cup \ldots \cup C_i) \sse C$ is clearly independent by
  subset-closure.  Since $\I$ is a $p$-independence system, $|C_1 \cup
  C_2 \cup \ldots \cup C_i| \leq p\cdot |S_i|$, and thus $\sum_i |C_i| =
  \sum_i p_i \leq i\cdot p$, proving the first property.

  For the second property, note that $C_i = A_{i-1} \setminus A_i \sse
  A_{i-1}$; hence each $e \in C_i$ does not belong to $S_{i-1}$ but
  could have been added to $S_{i-1}$ whilst maintaining independence,
  and was considered by the greedy algorithm. Since greedy chose the
  $e_i$ maximizing the ``gain'', $\delta_i \geq f_{S_{i-1}}(\{e\})$ for
  each $e \in C_i$. Summing over all $e \in C_i$, we get $p_i \delta_i
  \leq \sum_{e \in C_i} f_{S_{i-1}}(\{e\}) \leq f_{S_{i-1}}(C_i)$, where
  the last inequality is by the subadditivity of $f_{S_{i-1}}$. Again,
  by submodularity, $f_{S_{i-1}}(C_i) \leq f_{S_{k}}(C_i)$, which proves
  the second fact about the partition $\{C_j\}_{j=1}^k$ of $C$.
\end{proof}

Clearly, the greedy algorithm works no worse if we stop it when the best
``gain'' is negative, but the above proof does not use that fact.


\section{Proofs for Knapsack Constraints}
\label{sec:sviridenko}

The proof is similar to that in~\cite{Sviridenko04} and the proof of
\lref[Lemma]{lem:det-threshold}. We use notation similar
to~\cite{Sviridenko04} for consistency.  Let $f$ be a non-negative
submodular function with $f(\emptyset) = 0$.  Let $I = [n]$, and we are
given $n$ items with weights $c_i \in \bZ_+$, and $B \geq 0$; let $\F =
\{S \sse I \mid c(S) \leq B$, where $c(S) = \sum_{i \in S} c_i$. Our
goal to solve $\max_{S \sse \F} f(S)$.
To that end, we want to prove the following result:
\begin{theorem}
  There is a polynomial-time algorithm that outputs a collection of sets
  such that for any $C \in \F$, the collection contains a set $S$
  satisfying $f(S) \geq \frac12 f(S \cup C)$.~\footnote{A preliminary
    version of the paper claimed a factor of $(1-1/\rme)$ instead of
    $1/2$---we thank C.~Chekuri for pointing out the error.}
\end{theorem}

\subsection{The Algorithm}

The algorithm is the following: it constructs a polynomial number of
solutions and chooses the best among them (and in case of ties, outputs
the lexicographically smallest one of them).
\begin{itemize}
\item First, the family contains all solutions with cardinality $1,2,3$:
  clearly, if $|C| \leq 3$ then we will output $C$ itself, which will
  satisfy the condition of the theorem.

\item Now for each solution $U \sse I$ of cardinality $3$, we greedily
  extend it as follows: Set $S_0 = U$, $I_0 = I$. At step~$t$, we have a
  partial solution $S_{t-1}$. Now compute
  \begin{gather}
    \theta_t = \max_{i \in I_{t-1} \setminus S_{t-1}} \frac{f(S_{t-1} +
      i) - f(S_{t-1})}{c_i}.
  \end{gather}
  Let the maximum be achieved on index $i_t$. If $\theta_t \leq 0$,
  terminate the algorithm. Else check if $c(S_{t-1} + i_t) \leq B$: if
  so, set $S_t = S_{t-1} + i_t$ and $I_t = I_{t-1}$, else set $S_t =
  S_{t-1}$ and $I_t = I_{t-1} - i_t$. Stop if $I_t \setminus S_t =
  \emptyset$. 
\end{itemize}
The family of sets we output is all sets of cardinality at most three,
as well as for each greedy extension of a set of cardinality three, we
output all the sets $S_t$ created during the run of the algorithm. Since
each set can have at most $n$ elements, we get $O(n^4)$ sets output by
the algorithm.

\subsection{The Analysis}

Let us assume that $|C| = t > 3$, and order $C$ as $j_1, j_2,
\ldots, j_t$ such that
\begin{gather}
  j_k = \max_{j \in C \setminus \{j_1, \ldots, j_{k-1}\}} f_{\{j_1,
    \ldots, j_{k-1}\}}(\{j\}),
\end{gather}
i.e., index the elements in the order they would be considered by the
greedy algorithm that picks items of maximum marginal value (and does
not consider their weights $c_i$). Let $Y = \{j_1, j_2,
j_3\}$. Submodularity and the ordering of $C$ gives us the following:
\begin{lemma}
\label{lemma:sviri-1}
  For any $j_k \in C$ with $k \geq 4$ and any $Z \sse I \setminus
  \{j_1, j_2, j_3, j_k\}$, it holds that:
  \begin{align*}
    f_{Y \cup Z}(\{j_k\}) &\leq f(\{j_k\}) \leq f(\{j_1\}) \\
    f_{Y \cup Z}(\{j_k\}) &\leq f(\{j_1, j_k\}) - f(\{j_1\})\leq
    f(\{j_1, j_2\}) - f(\{j_1\}) \\
    f_{Y \cup Z}(\{j_k\}) &\leq f(\{j_1, j_2, j_k\}) - f(\{j_1, j_2\})\leq
    f(\{j_1, j_2, j_3\}) - f(\{j_1, j_2\})
  \end{align*}
\end{lemma}
Summing the above three inequalities we get that for $j_k \not\in Y \cup
Z$,
\begin{gather}
  \label{eq:ksviri-5}
  3\,f_{Y \cup Z}(\{j_k\}) \leq f(Y).
\end{gather}
For the rest of the discussion, consider the iteration of the algorithm
which starts with $S_0 = Y$. For $S$ such that $S_0 = Y \sse S \sse I$,
recall that $f_Y(S) = f(Y \cup S) - f(Y) = f(S) -
f(Y)$. \lref[Proposition]{prob:submodfacts} shows that $f_Y(\cdot)$ is a
submodular function with $f_Y(\emptyset) = 0$. The following lemma is
the analog of~\cite[eq.~2]{Sviridenko04}:
\begin{lemma}
  \label{lemma:sviri-2}
  For any submodular function $g$ and all $S, T \sse I$ it holds that
  \begin{gather}
    g(T \cup S) \leq g(S) + \sum_{i \in T \setminus S} (g(S + i) - g(S))
  \end{gather}
\end{lemma}
\begin{proof}
  $g(T \cup S) =  g(S) + (g(T \cup S) - g(S)) = g(S) + g_S(T \setminus
  S) \leq g(S) + \sum_{i \in T \setminus S} g_S(\{i\}) = g(S) + \sum_{i
    \in T \setminus S} (g(S + i) - g(S))$, where we used subadditivity
  of the submodular function $g_S$.
\end{proof}

Let $\tau + 1$ be the first step in the greedy algorithm at which either
(a) the algorithm stops because $\theta_{\tau + 1} \leq 0$, or (b) we
consider some element $i_{\tau + 1} \in C$ and it is dropped by the
greedy algorithm---i.e., we set $S_{\tau + 1} = S_{\tau}$ and
$I_{\tau+1} = I_{\tau} - i_{\tau+1}$. Note that before this point either
we considered elements from $C$ and picked them, or the element
considered was not in $C$. In fact, let us assume that there are no
elements that are neither in $C$ nor are picked by our algorithm,
since we can drop them and perform the same algorithm and analysis
again, it will not change anything---hence we can assume we have not
dropped any elements before this, and $S_t = \{i_1, i_2, \ldots, i_t\}$
for all $t \in \{0, 1, \ldots, \tau\}$.

Now we apply \lref[Lemma]{lemma:sviri-2} to the submodular function
$f_Y(\cdot)$ with sets $S = S_t$ and $T = C$ to get
\begin{gather}
  \label{eq:k10}
  f_Y(C \cup S_t) \leq f_Y(S_t) + \sum_{i \in C \setminus S_t}
  f_Y(S_t + i) - f_Y(S_t) = f_Y(S_t) + \sum_{i \in C \setminus S_t}
  f(S_t + i) - f(S_t)
\end{gather}
Suppose case~(a) happened and we stopped because $\theta_{\tau + 1} \leq
0$. This means that every term in the summation in~\eqref{eq:k10} must be
negative, and hence $f_Y(C \cup S_\tau) \leq f_Y(S_\tau)$, or
equivalently, $f(C \cup S_\tau) \leq f(S_\tau)$. In this case, we are
not even losing the $(1 - 1/\rme)$ factor.

Case~(b) is if the greedy algorithm drops the element $i_{\tau + 1} \in
C$. Since $i_{\tau + 1}$ was dropped, it must be the case that
$c(S_\tau) \leq B$ but $c(S_\tau + i_{\tau + 1}) = B' > B$. In this case
the right-hand expression in~\eqref{eq:k10} has some positive terms for
each of the values of $t \leq \tau$, and hence for each $t$, we get
\begin{gather}
  \label{eq:k11}
  f_Y(C \cup S_t) \leq f_Y(S_t) + B \cdot \theta_{t+1}.
\end{gather}
To finish up, we prove a lemma similar to \lref[Lemma]{lem:det-threshold}.
\begin{lemma}
  \label{lemma:sviri-3}
  $f_Y(S_\tau + i_{\tau +1}) \geq \frac12\; f_Y(S_\tau \cup C)$.
\end{lemma}

\begin{proof}
  If not, then we have 
  \[ f_Y(S_\tau \cup C) - f_Y(S_\tau + i_{\tau + 1}) > f_Y(S_\tau +
  i_{\tau + 1}).
  \] 
  Since we are in the case that $\theta_{\tau + 1} > 0$, we know that
  $f_Y(S_\tau + i_{\tau + 1}) > f_Y(S_\tau)$, and hence 
  \[ f_{Y \cup S_{\tau}}(C) = f_Y(S_\tau \cup C) - f_Y(S_\tau) >
  f_Y(S_\tau + i_{\tau + 1}). 
  \] 
  Now, the subadditivity of $f_{Y \cup S_{\tau}}()$ implies that there
  exists some element $e \in C$ with $\frac{f_{Y \cup S_{\tau}}(e)}{c_e}
  > \frac{f_Y(S_\tau + i_{\tau + 1})}{B}$. Submodularity now implies
  that at each point in time $i \leq \tau + 1$, the marginal increase
  per unit cost for element $e$ is $\frac{f_{Y \cup S_i}(e)}{c_e} >
  \frac{f_Y(S_\tau + i_{\tau + 1})}{B}$. Now since the greedy algorithm
  picked elements with the largest marginal increase per unit cost, the
  marginal increase per unit cost at each step was strictly greater than
  $\frac{f_Y(S_\tau + i_{\tau + 1})}{B}$. Hence, at the moment the total
  cost of the picked exceeded $B$, the total value accrued would be
  strictly greater than $f_Y(S_\tau + i_{\tau + 1})$, which is a
  contradiction.
\end{proof}

Now for the final calculations:
\begin{align*}
  f(S_\tau) &\geq f(Y) + f_Y(S_\tau) \\
  &\geq f(Y) + f_Y(S_\tau + i_{\tau + 1}) - \big(f_Y(S_\tau + i_{\tau +
    1}) - f_Y(S_\tau)\big) \\
  &\geq f(Y) + f_Y(S_\tau + i_{\tau + 1}) - \big(f(S_\tau + i_{\tau +
    1}) - f(S_\tau)\big) \\
  &\geq f(Y) + (1/2)f_Y(S_\tau \cup C) - f(Y)/3 \tag{using
    \lref[Lemma]{lemma:sviri-3} and~\eqref{eq:ksviri-5}}\\
  &\geq (1/2)f(S_\tau \cup C). \tag{by the definition of $f_Y()$}
\end{align*}
Hence this set $S_\tau$ will be in the family of sets output, and will
satisfy the claim of the theorem.

\section{Proofs from the Submodular Secretaries Section}
\label{sec:secretariesproofs}

In this section, we give the missing proofs from
\lref[Section]{sec:secretaries}.

\subsection{Proof for Cardinality Constrained Submodular Secretaries}

\textbf{\lref[Theorem]{thm:conversion}} \emph{The algorithm for the
  cardinality-constrained submodular maximization problem in the
  secretary setting gives an $O(1)$ approximation to $\OPT$.}

\medskip\noindent The proof basically shows that with reasonable
probability, both the first and the second half of the stream have a
reasonable fraction of $\OPT$, so when we run the offline algorithm on
the first half, using its output to extract value from the second half
gives us a constant fraction of $\OPT$.

\begin{proof}
  Let $C^* = \{e_1,\ldots,e_{k'}\}$ denote some set with $k' \leq k$
  elements such that $f(C^*) = \OPT$. Without loss of generality, we
  normalize so that $\OPT = 1$. Suppose the elements of $C^*$ have been
  listed in the ``greedy order'' (i.e., in order of decreasing marginal
  utility), and let $a_i$ denote the marginal utility of $e_i$ when it
  is added to $\{e_1, e_2, \ldots, e_{i-1}\}$. We consider two cases: in
  the first case, $a_1 \geq 1/c$, where $c \geq 1$ is some constant to
  be determined. In this case, with probability $1/2e$, the algorithm
  runs Dynkin's secretary algorithm and selects $a_1$, achieving an
  $1/(2ce)$ approximation.

  In the other case, $a_i < 1/c$ for all $i$. We imagine randomly
  partitioning the elements of the input set $X$ into two sets, $X_1$
  and $X_2$, with each element belonging to $X_1$ independently with
  probability $1/2$. This corresponds to the algorithm's division of
  $\stream$ into the first (random) $m$ elements $\stream_m$ and the
  remaining elements $\stream-\stream_m$.  Let $C_1^*$ and $C_2^*$
  denote the optimal solutions restricted to sets $X_1$ and $X_2$
  respectively. Define the random variable $A =
  \sum_{i=1}^{k'}Y_ia_i$ where each $Y_i \in_r\{-1,1\}$ is selected
  uniformly at random. Note that by submodularity, $f(C_1^*) + f(C_2^*)
  \geq f(C^*) = 1$. We wish to lower bound $\min(f(C_1^*),f(C_2^*))$,
  and to do this it is sufficient to upper bound the absolute value
  $|A|$. To see this, suppose that, for some setting of the
  $Y_i$'s it holds that $\sum_{i: Y_i = 1} a_i \geq \sum_{i: Y_i = -1}
  a_i$ (the other case is identical). Now if $|A| = \sum_{i:
    Y_i = 1} a_i - \sum_{i : Y_i = -1}a_i \leq x$, we have:
  \begin{gather*}
    \sum_{i : Y_i = -1}a_i \geq (\sum_{i: Y_i = 1} a_i )  -x  = 1- (\sum_{i : Y_i = -1}a_i) - x
  \end{gather*}
  and hence
  \begin{gather*}
    \min(f(C_1^*),f(C_2^*)) \geq \sum_{i : Y_i = -1}a_i  \geq \frac{1-x}{2}.
  \end{gather*}
  Hence, we would like to upper bound $|A|$ with high probability.
  Since each $Y_i$ is independent with expectation~$0$, we have $\E[A] =
  0$ and $\E[A^2] = \sum_{i=1}^{k'}a_i^2$. The standard deviation of $A$
  is:
  \begin{gather*}
    \sigma = \sqrt{\sum_{i=1}^{k'}a_i^2} \leq \sqrt{c\cdot \frac{1}{c}^2} = \frac{1}{\sqrt{c}}.
  \end{gather*}
  By Chebyshev's inequality, for any $d \geq 0$, we have
  \begin{gather*}
    \Pr[|A| \geq \frac{d}{\sqrt{c}}] \leq \frac{1}{d^2}.
  \end{gather*}
  That is, except with probability $1/d^2$, $\min(f(C_1^*),f(C_2^*))
  \geq (1-\frac{d}{\sqrt{c}})/2$.  Now for some calculations. With
  probability $1/2$, we do not run Dynkin's algorithm. Independently of
  this, with probability $1/2$, $f(C_1^*) \leq f(C_2^*)$---i.e., the
  value $\min(f(C_1^*),f(C_2^*))$ is achieved on $\stream_m$. With
  probability $(1-1/d^2)$, this value is at least
  $(1-\frac{d}{\sqrt{c}})/2$. Now we run a
  $\rho_{\text{off}}$-approximation on $\stream_m$, and thus with
  probability $\frac14 (1- 1/d^2)$,
  \begin{gather*}
    f(A_1) \geq \frac12 (1-\frac{d}{\sqrt{c}}) \cdot \frac{1}{\rho_{\text{off}}}.
  \end{gather*}
  If we use this as a lower bound for $f(C_2^*)$ (which is fine since we
  are in the case where $f(A_1) \leq f(C_1^*) \leq f(C_2^*)$), the
  semi-online algorithm gives us a value of at least
  $\frac{f(A_1)}{\rho_{on}}$. Hence we have
  \begin{gather}
    \label{eq:13}
    \E[f(A_2)]  \geq \frac12 (1-\frac{d}{\sqrt{c}}) \cdot
    \frac{1}{\rho_{\text{off}}} \cdot \frac14 (1- 1/d^2) \cdot
    \frac{1}{\rho_{\text{on}}}.
 \end{gather}
 Combining both cases and optimizing over parameters $d$ and $c$ ($d
 \leftarrow 3.08,\ c \leftarrow 260.24)$ we have:
 \begin{gather}
   \label{eq:14}
   \E[f(A_2)] \geq \min\left(\frac{1}{8\, \rho_{\text{off}}
       \,\rho_{\text{on}}}
     (1-1/d^2)(1-\frac{d}{\sqrt{c}}),\frac{1}{2ce}\right)\cdot\OPT \geq
   \frac{\OPT}{1417}
 \end{gather}
\end{proof}


\subsection{Proof for Partition Matroid Submodular Secretaries}

Let $S_0$ be the set of first $N_0$ elements, and let $S_j$ denote the elements in epoch $j$. Since the input permutation itself is random, the distribution over the sets $S_0,\ldots,S_k$ is identical to one resulting from the following process: each element $e$ independently chooses a real number $r_e$ in $(0,1)$ and is placed in $S_0$ if $r_e \leq \frac{1}{2}$, and in $S_j$ if $r_e \in (\frac{1}{2}+\frac{j-1}{100k},\frac{1}{2}+\frac{j}{100k}]$. We shall use this observation to simplify our analysis.

For the following lemma, we need to keep track of several events:
\begin{enumerate}
\item $H_{i,j}$: The highest element from partition group $i$ defined under the valuation function used during epoch $j$ falls into epoch $j$.
\item $F_{i,j}$: The highest element from partition group $i$ among those seen until the end of epoch $j$ (defined under the valuation function used during epoch $j$) falls into epoch $j$.
\item $L_{i,j}$: The highest element from partition group $i$ defined under the valuation function used during epoch $j$ does not fall before epoch $j$.
\item $S_{i,j}$: The second highest element (if any) from partition group $i$ defined under the valuation function used during epoch $j$ falls before epoch  $j$.
\item $P_{i,j}$: Some element from partition group $i$ has already been selected before epoch $j$.
\end{enumerate}
In the definitions above, we assume that a fixed tie breaking rule is used to ensure that there is a unique highest and second highest element.

\textbf{Lemma \ref{lem:partitionTechnicalLemma}}
{\it
For all partition groups $i$ and epochs $j$, the algorithm selects the highest element from group $i$ (according to the valuation function during epoch $j$) during epoch $j$ with probability at least $\Omega(\frac{1}{k})$. Specifically:
$$\Pr[(H_{i,j} \wedge S_{i,j} \wedge \neg P_{i,j}) \wedge (\bigwedge_{i' \ne i}\neg F_{i',j})] = \Omega(\frac{1}{k})$$
where the probability is over the random permutation of the elements.}

\begin{proof}
We observe that the event $(H_{i,j} \wedge S_{i,j} \wedge \neg P_{i,j}) \wedge (\bigwedge_{i' \ne i}\neg F_{i',j})$ implies that algorithm selects the highest element from group $i$ in epoch $j$. We will lower bound the probability of this event. We will show this by considering the events $P_{i,j},S_{i,j},L_{i,j},\bigwedge_{i' \ne i}\neg F_{i',j}, H_{i,j}$ in this order, and lower bound the probability of each conditioning on the previous ones.

Under any (arbitrary) valuation function, the events $L_{i,j}$ and $S_{i,j}$ depend on the real numbers chosen by the highest and the second highest elements. Thus $\Pr[L_{i,j} \wedge S_{i,j}] \geq \frac{1}{2}(\frac{1}{2}-\frac {j-1}{100k}) \geq \frac{1}{5}$. 

Let $Q_{i,j}$ denote the number of elements from group $i$ that do not appear in $S_0,\ldots,S_{j-1}$, but are higher (under the valuation function at epoch $j$) than any group $i$ element in $S_0,\ldots,S_{j-1}$. It is easy to see that the random variable $Q_{i,j}$ is dominated by a geometric random variable with parameter $\frac{1}{2}$. Moreover, for any element $e$ contributing to $Q_{i,j}$, it appears in epoch $j$ with probability at most $\frac{\frac{1}{100k}}{\frac{1}{2}-\frac{j}{100k}} \geq \frac{1}{40k}$ so that $Pr[F_{i,j}] \leq \frac{E[Q_{i,j}]}{40k} \leq \frac{1}{20k}$. Since $P_{i,j} \subseteq \cup_{j'<j} F_{i,j}$, we conclude that  $\Pr[P_{i,j}] \leq \sum_{j' < j}\Pr[F_{i,j}] \leq \frac{1}{20}$. It follows that $\Pr[L_{i,j} \wedge S_{i,j} | \neg P_{i,j}] \geq \Pr[L_{i,j} \wedge S_{i,j}] - \Pr[P_{i,j}] \geq \frac{1}{5} - \frac{1}{20} = \frac{3}{20}$.

For convenience, let us define event $\mathcal{E}_{i,j} = (L_{i,j} \wedge S_{i,j} \wedge \neg P_{i,j})$. We have:
$$\Pr[\mathcal{E}_{i,j}] = \Pr[L_{i,j} \wedge S_{i,j} | \neg P_{i,j}]\Pr[\neg P_{i,j}] \geq \frac{3}{20}(1-\frac{1}{20}) = \frac{57}{400}$$

We next upper bound the probability that groups $i' \ne i$ have elements in epoch $j$ that the algorithm might select, conditioned on $\mathcal{E}_{i,j}$. With $Q_{i',j}$ defined as above, we have 
$$\Pr[F_{i',j} | \mathcal{E}_{i,j}] \leq E[Q_{i',j} | \mathcal{E}_{i,j}] \cdot \frac{1}{40k}
 \leq \frac{E[Q_{i',j}]}{\Pr[\mathcal{E}_{i,j}]}\cdot \frac{1}{40k} \leq  \frac{20}{57k}.$$ 
 
Since there are at most $k$ groups $i'$, by a union bound: $\Pr[\bigvee_{i' \ne i} F_{i',j} | \mathcal{E}_{i,j}] \leq \frac{20}{57}$, and so: $\Pr[\bigwedge_{i' \ne i} \neg F_{i',j} | \mathcal{E}_{i,j}] \geq \frac{37}{57}$. Consequently:
$$\Pr[\mathcal{E}_{i,j} \wedge (\bigwedge_{i' \ne i} \neg F_{i',j})] \geq  \Pr[\mathcal{E}_{i,j}]\cdot \Pr[\bigwedge_{i' \ne i} \neg F_{i',j} | \mathcal{E}_{i,j}]\geq \frac{57}{400}\cdot \frac{37}{57} = \frac{37}{400}.$$

To complete the proof, we observe $\Pr[H_{i,j} | \mathcal{E}_{i,j} \wedge (\bigwedge_{i' \ne i}\neg F_{i',j})] \geq \frac{1/100k}{2/5} = \frac{1}{40k}$ and so:
$$\Pr[H_{i,j} \wedge \mathcal{E}_{i,j} \wedge (\bigwedge_{i' \ne i}\neg F_{i',j})] \geq \frac{37}{400}\cdot\frac{1}{40 k} = \frac{37}{16000k}.$$
\end{proof}


\section{Lower Bounds for the Constrained Submodular Maximization Problem in the Secretary Setting}
\label{sec:lower-bounds}

In this section we show lower-bounds for the secretary problem over submodular functions.  We first note that Kleinberg~\cite{Kleinberg-multiple} showed that for additive functions, the maximization problem in the on-line setting with a $k$-uniform matroid constraint can be approximated within a factor of $1 - \frac{5}{\sqrt{k}}$.  We show that this is not the case for submodular functions, even in the information theoretic, semi-online setting (where the algorithm knows the value of OPT) by exhibiting a gap for arbitrarily large $k$.

\begin{theorem}  No   algorithm approximates submodular maximization in the semi-online setting with a $k$-uniform matroid constraint better than a factor of $\frac{8}{9}$ for $k = 2$ or $\frac{17}{18}$ for any even  $k$.\end{theorem}

No non-trivial bound is possible for $k = 1$ because the algorithm knows OPT.  Thus the standard secretary lower bounds will not work.

Let $R, S$,  be two finite sets such that $S \subseteq R$.  We define the {\sc cover$(R, S)$} as follows: define the universe to be  $U=\{ij: i \in R, j \in \{B, T\} \}$, define the set of elements $W$ to contain $i_B = \{iB\}$ for $i \in  R$ and $i_{TB}$ for $i = S$.  Define a submodular function $f(C) = |\bigcup_{S \in C} S|$ where $C \subseteq W$.

We first prove the case for $k = 2$ with a small example and case analysis.  Consider {\sc cover$(\{1, 2\}, \{r\})$}, where $r \in \{1, 2\}$.  The universe is $U = \{1B, 1T, 2B, 2T \}$.  The three elements are $1_B = \{1B\}$, $2_B = \{2B\}$ and $r_{TB} = \{iB, iT \}$.

\begin{figure}[htb]
  \begin{center}
    \includegraphics[angle=0,width=3.0in]{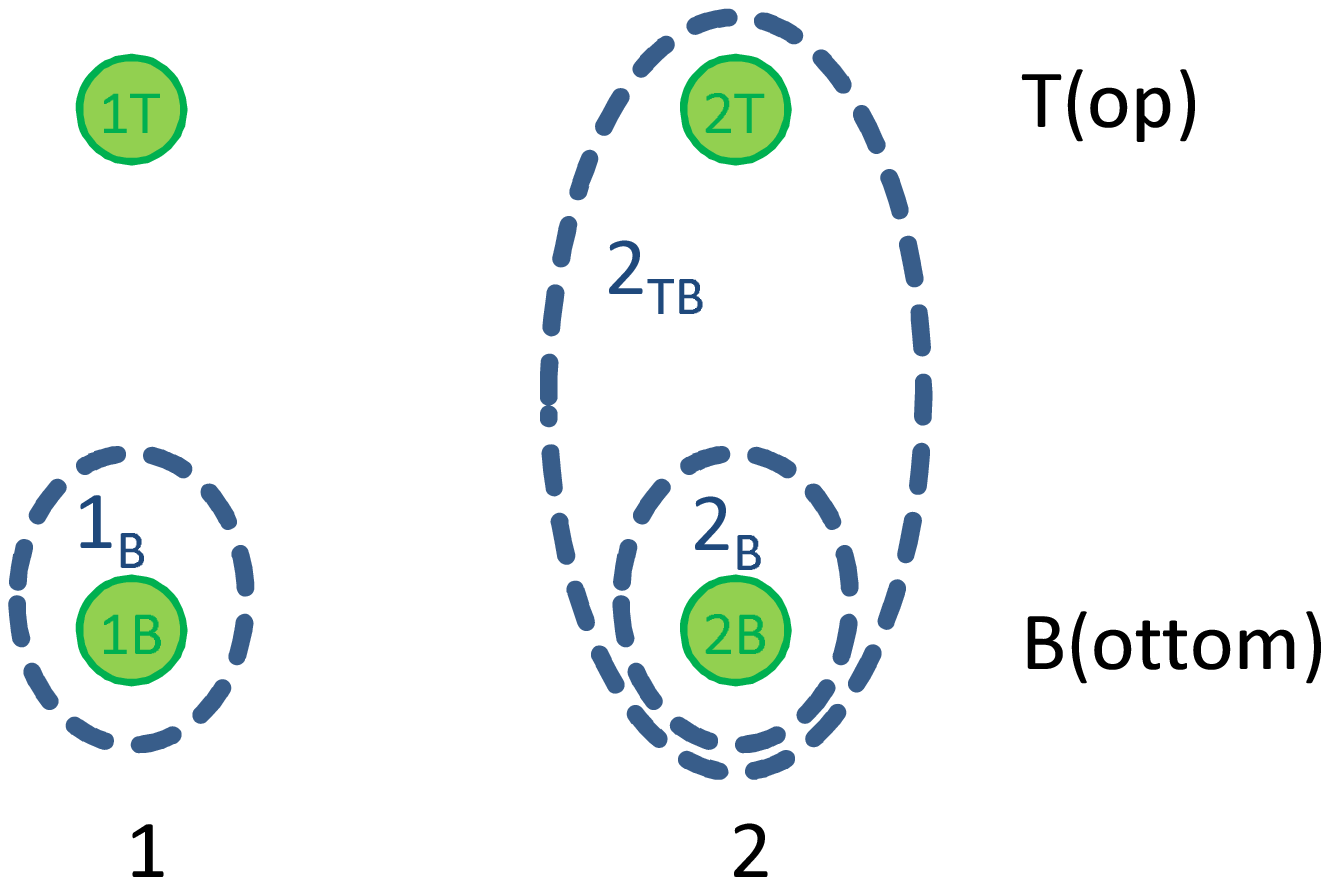}
  \end{center}
  \caption{\small Illustration of {\sc cover$(\{1, 2\}, \{2\})$}}
  \label{fig-label}
\end{figure}

We will chose a uniformly random $r \in \{1, 2\}$ and in the semi-online setting will require the algorithm to pick at most $k = 2$ of the sets appearing in random order, while trying to maximize $f$.  Let $\bar{r} = 3-r$, then the offline OPT is $C^*=\{r_{TB}, \bar{r}_{B} \}$ with $f(C^*)=3$

\begin{claim} \label{'claim-2k'}  No algorithm has expected payoff greater than $\frac{8}{3}$ on the instance {\sc cover$(\{1, 2\}, \{r\})$} in the semi-online setting when $r$ is drawn uniformly at random.
\end{claim}

Because $OPT = 3$, Claim~\ref{'claim-2k'} implies  no algorithm can do better than $\frac{8}{9}$ fraction of $OPT$, which gives us the first part of the theorem.

\begin{proof}
We proceed by case analysis.  In the case where the first element that arrives is $r_{TB}$, the algorithm knows $r$ and can obtain $OPT = 3$.  This happens with probability $\frac{1}{3}$.

In the case where the first element that arrives is $1_B$, the algorithm can accept or reject the element.  If the algorithm rejects, then it may as well take the next two elements that arrive. Since $r=1$ with probability half, the expected payoff is at most $\frac{5}{2}$.

Now suppose the algorithm accepts $1_B$ in the the first position. The algorithm should now pick $r_{TB}$ (and reject $2_B$ if it comes before $r_{TB}$) because the marginal value of $r_{TB}$ is at least as large as that of $2_B$. Since $r$ is random, this marginal value is $\frac{3}{2}$ in expectation, and hence the expected payoff of the algorithm is once again $\frac{5}{2}$.

Similarly, if the first element is $2_B$, the payoff is bounded by $\frac{5}{2}$ in expectation.
Thus the total expected payoff of the algorithm is bounded by
$\frac{1}{3}\cdot 3 + \frac{1}{3}\cdot \frac{5}{2} + \frac{1}{3}\cdot \frac{5}{2} = \frac{8}{3}$.

\end{proof}

Now we would like to show that something similar is true for much larger $k$.  The basic idea is to combine many disjoint instances of {\sc cover$(\{1, 2\}, \{r\})$}, and show that if the algorithm does well overall, it must have done well on each instance, violating Claim~\ref{'claim-2k'}.

\begin{claim} \label{'claim-k-big'}  For any even $k$, no algorithm has expected payoff more than $\frac{17}{12}k$ in the semi-online setting on instances of {\sc cover$(\{1, \ldots, k\}, S)$} trying to maximize $f$ and restricted to picking $k$ sets, when $S$ is drawn uniformly at random among subsets of $\{1, \ldots, k\}$ with $k/2$ elements.
\end{claim}

Because $OPT = 3k/2$., Claim~\ref{'claim-2k'} implies no algorithm can do better than a $\frac{17}{18}$ fraction of $OPT$, which gives us the second part of the theorem.

\begin{proof}
For the sake of analysis, we think of the instance of {\sc cover$(\{1, \ldots, k\}, S)$} being created by first choosing a matching on the set $\{1, \ldots, k\}$ and then within each edge $e = (i, j)$ of the matching choosing $e_r \in \{i, j\}$ to include in $S$.

We can then think of the sets of {\sc cover$(\{1, \ldots, k\}, S)$} being generated by taking the sets of the instance {\sc cover$(\{i, j\}, \{e_r\}$)} for each edge $e = (i, j)$ in the matching.  Call each of these $k/2$ instances of  {\sc cover$(\{i, j\}, \{e_r\}$)} a \emph{puzzle}.

Fix an semi-online algorithm $A$.  Let $C$ be the set of elements chosen by $A$.  For $0 \leq i \leq 3$, let $P_i$ be the set of puzzles such that $C$ contains exactly $i$ elements from the puzzle;  let $x_i$ be the expected sizes of  $P_i$ (over the randomness of the assignments of puzzles, the ordering, and $A$); and let $E_i$ be the expected payoff from all the puzzles in $P_i$.  Note that $E_0 = 0$ and $E_3 = 3 x_3$.

\begin{claim}  $E_1 + E_2 \leq \frac{4k}{3} - 2 x_3$ \end{claim}


\begin{proof}
Given an instance of  {\sc cover$(\{1,2\}, \{r\})$}, construct a random instance of {\sc cover$(\{1, \ldots, k\}, S)$} by generating a random matching and randomly picking a special edge $e = (i, j)$, where $i$ and $j$ are randomly ordered.  For each edge $e' = (i', j')$ pick $e'_{r} \in \{i', j'\}$ to include in $S$.  Now run $A$ on this instance of {\sc cover$(\{1, \ldots, k\}, S)$}, except than whenever an element of the puzzle corresponding to edge $e$ comes along, replace it with the next element from the given instance of {\sc cover$(\{1,2\}, \{r\})$}; however replace $1$ with $i$, and $2$ with $j$.  Run $A$ on this instance of {\sc cover$(\{1, \ldots, k\}, S)$}, and wheneven $A$ chooses an element from the instance of {\sc cover$(\{1,2\}, \{r\})$}, choose that element (it may be that $A$ selects more than $2$ elements, in which case, just select the first 2).

This instance of {\sc cover$(\{1, \ldots, k\}, S)$} has the same distribution as in the claim,
and the given instance of  {\sc cover$(\{1,2\}, \{r\})$}
is a random puzzle in this distribution.
Thus, the expected payoff of the  {\sc cover$(\{1,2\}, \{r\})$}
instance is at least $(E_0 + E_1 + E_2+ 2 x_3)/(\frac{k}{2})$.
By Claim~\ref{'claim-2k'} this is $\leq \frac{8}{3}$. Recalling $E_0 = 0$ and rearranging gives us the claim.
\end{proof}

We combing the above claim with the fact that $E_0 = 0$ and $E_3 = 3 x_3$ to get that \begin{equation} \label{'equation-lb-1'} \E[f(C)] = E_1 + E_2 + E_3  \leq \frac{4}{3}k + x_3.  \end{equation}

Note also that $A$ receives payoff  at most 2 from any puzzle in $P_1$, and at most 0
from any puzzle in $P_0$.  The maximum payoff from each puzzle is 3, which occurs in OPT.  Thus  \begin{equation}\label{'equation-lb-2'} \E[f(C)] \leq 3k/2 - 3 x_0 - x_1. \end{equation}

Finally, there are $k/2$ puzzles, so $x_0+x_1+x_2+x_3 = k/2$.  Additionally, because the algorithm never picks more than $k$ elements, we have $x_1+2x_2+3x_3 \leq k$.  Solving the first equation for $x_2$ and substituting for $x_2$ in the second we get \begin{equation} \label{'equation-lb-3'} 0 \leq 2 x_0 + x_1 - x_3. \end{equation}

Adding Equations \ref{'equation-lb-1'},  \ref{'equation-lb-2'}, and  \ref{'equation-lb-3'}, we see that $2 \E[f(C)] \leq \frac{17}{6}k -  x_0$ which implies that $ \E[f(C)] \leq \frac{17}{12}k$.
\end{proof}

\end{document}